\documentclass[11pt,a4paper,reqno]{amsart}%
\usepackage{amsthm,amsmath,amsfonts,amssymb,amsxtra,appendix,bookmark,euscript,delarray,dsfont}
\usepackage[latin1]{inputenc}

\theoremstyle{plain}
\newtheorem{theorem}{Theorem}
\newtheorem{lemma}[theorem]{Lemma}

\theoremstyle{remark}

\newcommand\R{{\ensuremath {\mathbb R} }}

\newcommand\1{{\ensuremath {\mathds 1} }}

\newcommand\nn{\nonumber}

\renewcommand\phi{\varphi}

\newcommand{\gH}{\mathfrak{H}}

\newcommand{\wto}{\rightharpoonup}

\newcommand{\cE}{\mathcal{E}}

\newcommand{\eps}{\epsilon}

\renewcommand{\epsilon}{\varepsilon}

\newcommand{\norm}[1]{ \left| \! \left| #1 \right| \! \right| }
\DeclareMathOperator{\tr}{{\rm Tr}}
\DeclareMathOperator{\Tr}{{\rm Tr}}

\renewcommand{\ge}{\geqslant}
\renewcommand{\le}{\leqslant}
\renewcommand{\geq}{\geqslant}
\renewcommand{\leq}{\leqslant}
\renewcommand{\hat}{\widehat}

\title[2D focusing many-bosons systems]{A note on 2D focusing many-boson systems}

\author[M. Lewin]{Mathieu LEWIN}
\address{CNRS \& Universit\'e Paris-Dauphine, CEREMADE (UMR 7534), Place de Lattre de Tassigny, F-75775 PARIS Cedex 16, France} 
\email{mathieu.lewin@math.cnrs.fr}

\author[P.~T. Nam]{Phan Th\`anh NAM}
\address{IST Austria, Am Campus 1, 3400 Klosterneuburg, Austria} 
\email{pnam@ist.ac.at}

\author[N. Rougerie]{Nicolas ROUGERIE}
\address{Universit\'e Grenoble 1 \& CNRS,  LPMMC (UMR 5493), B.P. 166, F-38042 Grenoble, France}
\email{nicolas.rougerie@grenoble.cnrs.fr}

\begin{document}
\date{September 2015}

\begin{abstract} 
We consider a 2D quantum system of $N$ bosons in a trapping potential $|x|^s$, interacting via a pair potential of the form $N^{2\beta-1} w(N^\beta x)$. We show that for all $0<\beta<(s+1)/(s+2)$, the leading order behavior of ground states of the many-body system is described in the large $N$ limit by the corresponding cubic nonlinear Schr\"odinger energy functional. Our result covers the focusing case ($w<0$) where even the stability of the many-body system is not obvious. This answers an open question mentioned by X.~Chen and J.~Holmer for harmonic traps ($s=2$). Together with the BBGKY hierarchy approach used by these authors, our result implies the convergence of the many-body quantum dynamics to the focusing NLS equation with harmonic trap for all $0<\beta <3/4$. 
\end{abstract}

\maketitle

\setcounter{tocdepth}{2}
\tableofcontents

\section{Introduction}

Since the experimental realization of Bose-Einstein condensation (BEC) in dilute trapped Bose gases in 1995 \cite{Wieman-Cornell-95, Ketterle-95}, it has been an ongoing challenge in mathematical physics to derive the phenomenon from the first principles of quantum mechanics (see~\cite{BenPorSch-15,LieSeiSolYng-05,Rougerie-LMU} and references therein). 
The nature of the interaction between particles plays an essential role. In particular, singular and/or attractive potentials complicate the analysis dramatically. 
\medskip

In the present paper, we are interested in the derivation of the minimization problem for the 2D nonlinear Sch\"odinger (NLS) energy functional
\begin{equation}\label{eq:NLs func}
\cE_{\rm NLS}(u)=\int_{\R^2} \Big(  |(i\nabla +A(x)) u|^2 + V(x)|u(x)|^2 + \frac{a}{2} |u(x)|^4 \Big)dx 
\end{equation}
subject to the mass constraint
\begin{equation}\label{eq:mass cons}
\int_{\R ^2} |u| ^2 = 1.
\end{equation}
We will show that this NLS functional arises as an effective model for large dilute 2D bosonic systems, as a consequence of the occurence of BEC in the ground states. We shall be more specifically concerned with the focusing (or attractive) case, $a\leq 0$. 

Here $V$ is an external potential which serves to trap the system and $A$ is a vector potential corresponding to a magnetic field (or the effective influence of a rotation). We assume that
\begin{align} \label{eq:assumption-V}
V \in L^1_{\rm loc}(\R^2,\R), \quad A \in L^2_{\rm loc}(\R^2,\R^2)\quad \text{and} \quad V(x)\ge C^{-1} (|A(x)|^2 +|x|^s) - C
\end{align}
for a fixed parameter $s>0$ (we always denote by $C$ a generic positive constant whose value alters from line to line). The case $s=2$ corresponds to the harmonic trap which is most often used in laboratory experiments. 

We will assume that $a>-a^*$ where $a^*>0$ is the critical interaction strength for the existence of a ground state for the focusing NLS functional \cite{Weinstein-83,Zhang-00,GuoSei-14,Maeda-10}. In fact, $a^*$ is the optimal constant of the Gagliardo-Nirenberg inequality:
\begin{align} \label{eq:GN}
 \left(\int_{\R^2} |\nabla u|^2 \right) \left(\int_{\R^2} |u|^2 \right) \ge \frac{a^*}{2} \int_{\R^2} |u|^4.
\end{align}
Equivalently,  
$$a ^* = \norm{Q}_{L ^2 (\R ^2)} ^2, $$ 
where $Q\in H ^1 (\R ^2)$ is the unique (up to translations) positive radial solution of 
\begin{equation}\label{eq:Q-GN}
-\Delta Q + Q - Q ^3 = 0 \mbox{ in } \R ^2.
\end{equation}

The linear many-body model for $N$ identical bosons we start from is described by the Hamiltonian
\begin{equation}\label{eq:HN}
H_N = \sum_{j=1} ^N  \left( \left(i\nabla_{j} +A (x_j) \right) ^2 + V(x_j) \right) + \frac{1}{N-1} \sum_{1\leq i<j \leq N}w_N(x_i-x_j)
\end{equation}
acting on $\gH^N = \bigotimes_{\rm sym}^N L^2(\R^2)$, the Hilbert space of square-integrable symmetric functions. The two-body interaction is chosen of the form   
\begin{align} \label{eq:assumption-wN}
w_N(x)=N^{2\beta} w(N^\beta x)
\end{align}
for a fixed parameter $\beta>0$ and a fixed function $w$ satisfying
\begin{align}
\label{eq:assumption-w1} 
w, \hat w  \in L^1(\R^2,\R),  \quad  w(x)=w(-x) \quad \text{and}\quad \int_{\R ^2} w =a.
\end{align}
The coupling constant $1/(N-1)$ ensures that the total kinetic and interaction energies are comparable, so that we can expect a nontrivial effective theory in the limit $N\to \infty$. 
\medskip

Roughly speaking, BEC occurs when almost all particles live in a common quantum state, that is, in terms of wave functions, 
$$ \Psi(x_1,...,x_N) \approx u^{\otimes N}(x_1,...,x_N):=u(x_1)u(x_2)...u(x_N)$$
in an appropriate sense. By simply taking the trial wave functions $u^{\otimes N}$, we obtain the Hartree energy functional
\begin{align}\label{eq:Hartree func} 
\cE_{{\rm H},N}(u)&=\frac{\langle u^{\otimes N}, H_N u^{\otimes N}\rangle }{N} \nonumber\\
&= \int_{\R^2} \Big(  |(i\nabla u(x) + A(x)u(x)|^2 + V(x)|u(x)|^2 + \frac{1}{2} |u(x)|^2 (w_N*|u|^2)(x) \Big) d x.
\end{align}
The infimum of the latter, under the mass constraint $\int |u|^2=1$, is thus an upper bound to the many-body ground state energy per particle. When $N\to \infty$, since
\begin{align} \label{eq:wN-wto-delta}
w_N \wto \left(\int_{\R ^2} w\right) \delta_0 = a \delta_0,
\end{align}
the Hartree functional \eqref{eq:Hartree func} {\em formally} boils down to the NLS functional \eqref{eq:NLs func}. On the other hand, the Hartree functional is stable in the limit $N\to \infty$ only if 
\begin{equation}
\inf_{u\in H^1(\R^2)}\left(\frac{\displaystyle\iint_{\R^2\times \R^2}|u(x)|^2|u(y)|^2w(x-y)\,dx\,dy}{2\norm{u}_{L^2(\R^2)}^2\norm{\nabla u}_{L^2(\R^2)}^2}\right) \ge -1.
\label{eq:Hartree stable non strict}
\end{equation}
In fact, if \eqref{eq:Hartree stable non strict} fails to hold, then the ground state energy of the Hartree functional converges to $-\infty$ as $N\to \infty$, see \cite[Prop. 2.3]{LewNamRou-14c}. Hence, Condition \eqref{eq:Hartree stable non strict} is {\em necessary} for the many-body Hamiltonian to satisfy stability of the second kind: 
\begin{equation}
\label{eq:HN>-CN} H_N \ge -C N.
\end{equation}
That the one-body stability condition~\eqref{eq:Hartree stable non strict} is also \emph{sufficient} to ensure~ the many-body stability \eqref{eq:HN>-CN} is highly nontrivial and it is one of the main concerns of the present paper. As in \cite{LewNamRou-14c}, we will actually assume the strict stability
\begin{align} \label{eq:assumption-w2}
\inf_{u\in H^1(\R^2)}\left(\frac{\displaystyle\iint_{\R^2\times \R ^2}|u(x)|^2|u(y)|^2w(x-y)\,dx\,dy}{2\norm{u}_{L^2(\R^2)}^2\norm{\nabla u}_{L^2(\R^2)}^2}\right)>-1
\end{align} 
which plays the same role as the assumption $a>-a^*$ in the NLS case. Note that \eqref{eq:Hartree stable non strict} implies that $\int w \ge -a^*$, and \eqref{eq:Hartree stable non strict} holds if $\int_{\R^2}|w_-|<a^*$.
\medskip

The goal of the present paper is to improve on the results of~\cite{LewNamRou-14} where we showed in particular that  the many-body ground states converge (in terms of reduced density matrices) to those of the NLS functional~\eqref{eq:NLs func} when $N\to \infty$, provided 
$$0 < \beta < \beta_0 (s) := \frac{s}{4(s+1)}.$$
Here we extend this range to 
\begin{equation}\label{eq:beta condition}
\boxed{0 < \beta < \beta_1 (s):= \frac{s+1}{s+2}.}
\end{equation}
Note the qualitative improvement: while $\beta_0(s) <1/2$, we have $\beta_1 (s) > 1/2$. This means that we now allow the range of the interaction to be much smaller than the typical distance between particles, of order $N^{-1/2}$. We can thus treat a dilute limit where interactions are rare but strong, as opposed to the previous result which was limited to the mean-field case with many weak interactions.  
%

\medskip

\noindent{\bf Acknowledgements.} The authors acknowledge financial support from the European Union's Seventh Framework Programme (ERC Grant MNIQS no. 258023 and REA Grant no. 291734) and the ANR (Mathostaq Project ID ANR-13-JS01-0005-01).

\section{Main results}

\subsection{Statements}

We will prove the convergence of the ground state energy per particle of $H_N$ to that of the NLS functional \eqref{eq:NLs func}. These are denoted respectively by
\begin{equation}\label{eq:GS ener many}
e_N := N^{-1}\inf_{\Psi\in \gH^N, \|\Psi\|=1} \langle \Psi_N, H_N \Psi_N \rangle \quad \text{and}\quad e_{\rm NLS} := \inf_{\|u\|_{L^2}=1} \cE_{\rm NLS} (u).
\end{equation}
The convergence of ground states is formulated using $k$-particles reduced density matrices, defined for any $\Psi \in\gH^N$ by a partial trace $$\gamma_{\Psi}^{(k)}:= \Tr_{k+1\to N} |\Psi \rangle \langle \Psi|.$$
Equivalently, $\gamma_\Psi^{(k)}$ is the trace class operator on $\gH^k$ with kernel 
$$
\gamma_\Psi^{(k)} (x_1,...,x_k; y_1,...,y_k)= \int_{\R^{2(N-k)}} \overline{\Psi(x_1,...,x_k,Z)}  \Psi(y_1,...,y_k,Z) dZ.
$$  
Our main result is the following

\begin{theorem}[\textbf{Convergence to NLS theory}]\label{thm:cv-nls}\mbox{}\\
Assume that $V$, $A$, $w$ satisfy \eqref{eq:assumption-V}, \eqref{eq:assumption-w1} and \eqref{eq:assumption-w2}. Then, for every $0<\beta<(s+1)/(s+2)$, 
\begin{align} \label{eq:thm-cv-GSE}
\boxed{\lim_{N\to \infty} e_N = e_{\rm NLS} >-\infty.} 
\end{align}
Moreover, for any ground state $\Psi_N$ of $H_N$, there exists a Borel probability measure $\mu$ supported on the ground states of $\cE_{\rm NLS}(u)$ such that, along a subsequence, 
\begin{align} \label{eq:thm-cv-DM}
\boxed{\lim_{N \to \infty}\Tr \left| \gamma_{\Psi _{N}}^{(k)} - \int |u^{\otimes k} \rangle \langle u^{\otimes k}| d\mu(u) \right| =0,\quad \forall k\in \mathbb{N}.}
\end{align}
If $\cE_{\rm NLS}(u)$ has a unique minimizer $u_0$ (up to a phase), then for the whole sequence
\begin{align} \label{eq:thm-BEC}
\lim_{N \to \infty}\Tr \left| \gamma_{\Psi_{N} }^{(k)} - |u_0^{\otimes k} \rangle \langle u_0^{\otimes k}|  \right| =0,\quad \forall k\in \mathbb{N}.
\end{align}
\end{theorem}

Note that if $A=0$ and $V$ is radial, one can prove the uniqueness for the NLS ground state by well-known arguments, reviewed for instance in~\cite{Frank-13}. Uniqueness can certainly fail when $A\neq 0$ (due to the occurence of quantized vortices~\cite{Sei-02}), or when $a<0$ and $V$ has several isolated minima~\cite{AshFroGraSchTro-02,GuoSei-14}.


 \subsection{Focusing quantum dynamics}  Most recently, Chen and Holmer~\cite{CheHol-15} considered the derivation of the time-dependent 2D focusing NLS in a harmonic trap $V(x) = |x| ^2$ from many-body quantum dynamics. They proved that for all $0<\beta<1/6$, if the initial state $\Psi_N(0)$ condensates on $u(0)$ (in the sense of density matrices as in \eqref{eq:thm-BEC}), then for every time $t>0$, the evolved state $\Psi_N(t)=e^{-it\widetilde{H}_N}\Psi_N(0)$ with 
$$
\widetilde{H}_N = \sum_{j=1}^N  \left( -\Delta_{x_j} + |x|^2 \right) + \frac{1}{N-1} \sum_{1\leq i<j \leq N} N^{2\beta}w(N^\beta(x_i-x_j))
$$
condensates on the solution $u(t)$ to the time-dependent NLS equation 
$$
i \partial_t u(t) = (-\Delta+|x|^2+ a |u(t)|^2) u(t), \quad u_{|t=0}=u(0).
$$
Their approach is based on the BBGKY hierarchy method and the stability of the second kind~\eqref{eq:HN>-CN}, which has been established in \cite{LewNamRou-14c} for $0<\beta<1/6$. As discussed in \cite[Section~2.3]{CheHol-15}, their method actually allows to treat any $0<\beta < 3/4$, provided that the stability holds for this larger range of $\beta$, which they left as an open question. Theorem~\ref{thm:cv-nls} thus provides the needed stability estimate to extend the main result in ~\cite{CheHol-15} to any~$0<\beta < 3/4$.

Note that if $\beta<1/2$, the next order correction to the 2D focusing quantum dynamics can be obtained using the Bogoliubov approach~\cite{LewNamSch-14,NamNap-15}  (see \cite{BocCenSch-15} for the defocusing case).

\subsection{Strategy of proof} We shall compare the many-body ground state energy per particle $e_N$ to that of the Hartree functional~\eqref{eq:Hartree func} 
$$
e_{{\rm H},N} := \inf_{\|u\|_{L^2}=1} \cE_{{\rm H},N}(u)
$$
and then use that (see Appendix~\ref{sec:app})
$$\lim_{N\to \infty} e_{{\rm H},N}=e_{\rm NLS}.$$ 
The upper bound $e_N \le e_{{\rm H},N}$ can be obtained using trial states $u^{\otimes N}$, and the difficult part is the matching lower bound. 

The first ingredient of our proof of Theorem \ref{thm:cv-nls} is the following: 

\begin{lemma}[\textbf{First lower bound on the ground state energy}] \label{lem:GSE-1}\mbox{}\\
For any $\beta \geq 0$ we have, in the limit $N\to \infty$,
\begin{align}\label{eq:first low bound}
e_N &\ge \inf_{\|u\|_{L^2}=1} \int_{\R^2} \Big( |\nabla u(x)|^2 + V |u(x)|^2 + \frac{1}{2} |u|^2(w_N*|u|^2) \Big) dx - CN^{2\beta-1}\nonumber\\
&\geq e_{\rm NLS} ^0 - o(1) - C N ^{2\beta -1}.
\end{align}
Here $e_{\rm NLS} ^0$ denotes the NLS energy with $A\equiv 0$.
\end{lemma}

\begin{proof}
The first inequality is proved in~\cite[Section 3]{Lewin-ICMP}. The second follows from the analysis of the Hartree functional in Appendix~\ref{sec:app}. 
\end{proof}

When $\beta<1/2$ and $A\equiv 0$ (no magnetic field), Lemma~\ref{lem:GSE-1} implies immediately the convergence of the ground state energy \eqref{eq:thm-cv-GSE}. When either $\beta\ge 1/2$ or $A\not\equiv 0$, the proof of the convergence~\eqref{eq:thm-cv-GSE} is more involved. In particular, when $\beta>1/2$ and $w<0$, the stability of the second kind~\eqref{eq:HN>-CN} is not provided by~Lemma~\ref{lem:GSE-1}. 

The main novelty of the present paper is to obtain~\eqref{eq:HN>-CN} by a bootstrap procedure,  taking Lemma~\ref{lem:GSE-1} as a starting point. 
As in~\cite{LewNamRou-14c}, a major ingredient in our proof is a quantitative version of the quantum de Finetti theorem. 

\begin{lemma}[\textbf{Quantitative quantum de Finetti}]\label{thm:deF-measure}\mbox{}\\ 
Let $\Psi\in \gH^N=\bigotimes_{\rm sym}^N L^2(\R^2)$ and let $P$ be a finite-rank orthogonal projector with 
$$\dim (P\gH)=d<\infty.$$
There exists a positive Borel measure $d\mu_\Psi$ on the unit sphere $S P \gH$ such that 
\begin{equation}\label{eq:deF estimate}
\tr \left| \int_{SP\gH} |u^{\otimes 2} \rangle \langle u^{\otimes 2}| d\mu_\Psi (u) - P^{\otimes 2} \gamma_\Psi^{(2)} P ^{\otimes 2}\right| \leq  \frac{8d}{N}
\end{equation}
and 
\begin{align} \label{eq:int-muN}
\int_{SP\gH} d\mu_\Psi(u) \ge \big(\Tr(P\gamma_\Psi^{(1)}) \big)^2.
\end{align} 
\end{lemma}

%
\begin{proof} The first inequality \eqref{eq:deF estimate} is contained in ~\cite[Lemma 3.4]{LewNamRou-14c}. The second inequality ~\eqref{eq:int-muN} is established in the course of the proof of~\cite[Lemma 3.8]{LewNamRou-14c}.
\end{proof}
We will apply the above lemma with $P$ a spectral projector below an energy cut-off $L$ for the one-particle operator:
\begin{equation}\label{eq:cut-off}
\quad P:= \1(h\le L)\quad \text{with} \quad h:=(i\nabla + A(x))^2 + V(x). 
\end{equation}
Note that Assumptions~\eqref{eq:assumption-V} ensure that 
\begin{align} \label{eq:estimate-h}
h \ge C^{-1} (1-\Delta+|A(x)|^2+V(x))-C \ge C^{-1}(-\Delta + |x|^s) -C.
\end{align}
Therefore, we have a Cwikel-Lieb-Rosenblum type estimate (see \cite[Lemma 3.3]{LewNamRou-14c})
\begin{equation} \label{eq:estimate-d-dimP} 
d:=\dim (P\gH)\le C L^{1+2/s}.
\end{equation}
The first main improvement over our previous work~\cite{LewNamRou-14c} is a better way to control the error induced by using the finite-rank cut-off $P$. Using several Sobolev-type estimates on the interaction operator $w_N$ (Lemma~\ref{lem:interaction-operator} below), we obtain

\begin{lemma}[\textbf{Second lower bound on the ground state energy}]\label{lem:GSE-2}\mbox{}\\
Let $\beta >0$. For every $\delta\in (0,1/2)$, there exists a constant $C_\delta>0$ such that for all $N \ge 2$, $L\ge 1$ and for all wave functions $\Psi_N\in \gH^N:$
\begin{align}\label{eq:GSE-2}
\frac{\langle \Psi_N, H_N \Psi_N\rangle}{N} \ge e_{{\rm H}, N} &- C_\delta L^{1+\delta}  \frac{d}{N} \nonumber\\
& -  \frac{C_\delta}{L^{1/4-\delta/2}} \Tr\left(h \gamma_{\Psi_N}^{(1)}\right) ^{1/4-\delta/2} \Tr \left(h\otimes h \gamma_{\Psi_N}^{(2)}\right) ^{1/2+\delta}.
\end{align}  
\end{lemma}

Lemma~\ref{lem:GSE-2} provides a sharp lower bound to the ground state energy if we have a strong enough {\em a-priori} control of the error terms in the second line of~\eqref{eq:GSE-2}. This is the other important improvement of the present paper. 

\begin{lemma}[\textbf{Moments estimates}]\label{lem:GSE-3}\mbox{}\\
Let $0<\beta <1$ and $\Psi_N \in \gH^N$ be a ground state of $H_N$. For all $\eps\in (0,1)$ we have 
\begin{equation}\label{eq:one-two body bound}
\Tr\left(h \gamma_{\Psi_N}^{(1)}\right) \leq C\frac{1 + |e_{N,\eps}|}{\eps}\quad \text{and} \quad  \Tr \left(h\otimes h \gamma_{\Psi_N}^{(2)}\right) \leq C \left( \frac{1 + |e_{N,\eps}|}{\eps} \right) ^2
\end{equation}
where 
\begin{equation}\label{eq:ener epsilon}
e_{N,\eps}:=N^{-1} \inf_{\Psi\in \gH^N, \|\Psi\|=1}  \left\langle \Psi, \Big( H_N - \eps \sum_{j=1}^N h_j \Big) \Psi \right\rangle. 
\end{equation}
\end{lemma}

This is reminiscent of similar estimates used by Erd\"os, Schlein and Yau for the time-dependent problem~\cite{ErdYau-01,ErdSchYau-07,ErdSchYau-10}. Recently, related ideas were also adapted to the ground state problem in \cite{NamRouSei-15}. Note, however, that these previous applications were limited to the defocusing case $w\ge 0$. Then $|e_{N,\eps}|$ is clearly bounded independently of $N$ and the moments estimates above allow to derive the NLS theory for any~$\beta <1$ (an even larger range of $\beta$ can be dealt with when $A=0$, using the methods of~\cite{LieSeiYng-01}). 

In the focusing case, we do not obtain actual a priori bounds by Lemma~\ref{lem:GSE-3}, since the estimates depend on $|e_{N,\eps}|$, which is essentially of the same order of magnitude as $|e_N|$. The uniform bound $|e_{N,\eps}|\le C$ will be obtained by a bootstrap argument: Lemma~\ref{lem:GSE-1} provides the starting point, and then the bounds in Lemmas~\ref{lem:GSE-2} and~\ref{lem:GSE-3} can be improved step by step, provided~\eqref{eq:beta condition} holds. Once stability of the second kind is proved, the convergence of the ground state energy~\eqref{eq:thm-cv-GSE} follows immediately from Lemma~\ref{lem:GSE-2}. The convergence of density matrices~\eqref{eq:thm-cv-DM} is a consequence of the proof of~\eqref{eq:thm-cv-GSE} and the quantum de Finetti Theorem, just as in~\cite{LewNamRou-14c}.

\medskip

\noindent{\bf Organization of the paper.} We will prove Lemma~\ref{lem:GSE-3} in Section~\ref{sec:GSE-3}, then Lemma~\ref{lem:GSE-2} in Section \ref{sec:GSE-2}. The proof of the main Theorem~\ref{thm:cv-nls} is concluded in Section \ref{sec:GSE-4}. Appendix~\ref{sec:app} contains the needed estimate to pass to the limit in the Hartree functional.

\section{Moments estimates: Proof of Lemma \ref{lem:GSE-3}} \label{sec:GSE-3}

Since we can always add a constant to $V$ if necessary, from now on we will assume that $V\ge 1$, and hence $h:=(i\nabla + A(x))^2 + V(x) \ge 1$. We will need the following 


\begin{lemma}[\textbf{Operator bounds for two-body interactions}]\label{lem:interaction-operator}\mbox{}\\
For every $W\in L^1 \cap L^2(\R^2,\R)$, the multiplication operator $W(x-y)$ on $L^2((\R^2)^2)$ satisfies
\begin{align}
& |W(x-y)|\le C \|W\|_{L^2} h_x, \label{eq:W<=Delta}\\
 & |W(x-y)|\le C_\delta \|W\|_{L^1} (h_x h_y)^{1/2+\delta}, \quad \forall \delta>0, \label{eq:W<=Delta-Delta} \\
& \pm \Big( h_x W(x-y) + W(x-y) h_x \Big) \le C \|W\|_{L^2} h_x h_y. \label{eq:Delta-W+W-Delta<=Delta-Delta}
\end{align}
\end{lemma}

Lemma \ref{lem:interaction-operator} is the 2D analogue of~\cite[Lemma 3.2]{NamRouSei-15}. The proof is similar and we omit it for shortness. Now we come to the

\begin{proof}[Proof of Lemma \ref{lem:GSE-3}] Note that $C\ge e_{{\rm H},N} \ge e_N \ge e_{N,\eps}$, and hence 
$|e_N|\le C(1+|e_{N,\eps}|)$. Clearly 
\begin{equation} 
\label{eq:eq:Tr-hG1} H_{N,\eps}:=H_N - \eps \sum_{j=1}^N h_j \ge N e_{N,\eps}.
\end{equation}
Taking the expectation against $\Psi_N$ and using the definition of the one-body density matrix we obtain the first inequality in~\eqref{eq:one-two body bound} immediately. To obtain the second inequality in~\eqref{eq:one-two body bound}, we use the ground state equation 
$$H_N \Psi_N = Ne_N \Psi_N$$
to write
\begin{align} \label{eq:dGammah2-upper} & \frac{1}{N^2}\left\langle \Psi_N, \Big(\Big(\sum_{j=1}^N h_j\Big)H_N  + H_N \Big(\sum_{j=1}^N h_j\Big)  \Big) \Psi_N \right\rangle \nn \\
& \qquad \qquad\qquad\qquad = \frac{2e_N}{N}\left\langle \Psi_N, \sum_{j=1}^N h_j \Psi_N \right\rangle \le \frac{C(1+|e_{N,\eps}|)^2}{\eps}.
\end{align}
Now we are after an operator lower bound on
\begin{align} \label{eq:wjk-hi-0}
\frac{1}{N^2} & \Big(\sum_{j=1}^N h_j\Big)H_N  + \frac{1}{N^2} H_N \Big(\sum_{j=1}^N h_j\Big) = \frac{2}{N^2}\Big(\sum_{j=1}^N h_j 
\Big)^2 \nn\\
&+ \frac{1}{N^2(N-1)}\sum_{i=1}^N \sum_{j<k} (h_i w_N(x_j-x_k)+ w_N(x_j-x_k) h_i).
\end{align}
For every $i=1,2,...,N$, we have
\begin{align} \label{eq:wjk-hi}
\frac{1}{N-1}\sum_{i\ne j<k \ne i} w_N(x_j-x_k) &= H_{N,\eps} - (1-\eps) \sum_{j=1}^N h_j - \frac{1}{N-1}\sum_{j\ne i} w_N(x_i-x_j)\nn\\
&\ge Ne_{N,\eps} - \left( 1-\eps + \frac{N^\beta}{N} \right) \sum_{j=1}^N h_j 
\end{align}
where we have used $H_{N,\eps}\ge Ne_{N,\eps}$ and applied ~\eqref{eq:W<=Delta} to obtain $ w_N(x_i-x_j)\le C N^\beta h_j.$ Note that both sides of \eqref{eq:wjk-hi} commute with $h_i$. Therefore, we can multiply \eqref{eq:wjk-hi} with $h_i$ and then take the sum over $i$ to obtain 
\begin{align} \label{eq:wjk-hi-1}
&\frac{1}{N^2(N-1)}\sum_{i=1}^N \sum_{i\ne j<k \ne i} (h_i w_N(x_j-x_k)+w_N(x_j-x_k)h_i) \nn\\
&\ge \frac{2e_{N,\eps}}{N}\sum_{j=1}^N h_j - \frac{2}{N^2}\Big( 1-\eps + \frac{CN^\beta}{N} \Big) \Big( \sum_{j=1}^N h_j \Big)^2 .
\end{align}
On the other hand, for every $j\ne k$, by \eqref{eq:Delta-W+W-Delta<=Delta-Delta} we have 
\begin{align*}
h_j w_N (x_j-x_k)  + w_N(x_j-x_k)h_j  \ge -C N^\beta h_j h_k .
\end{align*}
Therefore,
\begin{align} \label{eq:wjk-hi-2}
\frac{1}{N^2(N-1)}\sum_{j \ne k} \Big(h_j w_N(x_j-x_k)+w_N(x_j-x_k)h_j \Big) \ge - CN^{\beta-3} \Big(\sum_{j=1}^N h_j \Big)^2.
\end{align}
Inserting \eqref{eq:wjk-hi-1} and \eqref{eq:wjk-hi-2} into \eqref{eq:wjk-hi-0}, we find the operator bound
\begin{align}\label{eq:HN2}
&\frac{1}{N^2}\Big(\sum_{j=1}^N h_j\Big)H_N + \frac{1}{N^2} H_N \Big(\sum_{j=1}^N h_j\Big) \nonumber\\
&\qquad\qquad\qquad \ge \frac{2}{N^2}\Big( \eps - \frac{CN^{\beta}}{N} \Big)  \Big( \sum_{j=1}^N h_j\Big)^2 - \frac{C(1+|e_{N,\eps}|)}{N} \sum_{j=1}^N h_j. 
\end{align}
Taking the expectation against $\Psi_N$ and using the first inequality in \eqref{eq:one-two body bound}, we get 
\begin{align} \label{eq:dGammah2-lower}
&\frac{1}{N^2}\left\langle \Psi_N, \Big(\Big(\sum_{j=1}^N h_j\Big)H_N + H_N \Big(\sum_{j=1}^N h_j\Big) \Big) \Psi_N \right\rangle \nn \\
&\qquad\qquad\qquad\ge  \frac{2}{N^2}\Big(\eps-\frac{N^{\beta}}{N} \Big) \left\langle \Psi_N, \Big( \sum_{j=1}^N h_j \Big)^2 \Psi_N \right\rangle - \frac{C(1+|e_{N,\eps}|)^2}{\eps}.
\end{align}
Putting \eqref{eq:dGammah2-upper} and \eqref{eq:dGammah2-lower} together, we deduce that
$$
\frac{2}{N^2}\Big(\eps-\frac{N^{\beta}}{N} \Big) \left\langle \Psi_N, \Big( \sum_{j=1}^N h_j \Big)^2 \Psi_N \right\rangle \le \frac{C(1+|e_{N,\eps}|)^2}{\eps}.
$$
If $\beta<1$, then $\eps-CN^{\beta-1} \ge \eps/2>0$ for large $N$. Therefore, we conclude 
\begin{equation}\label{eq:moment final}
\frac{1}{N^2} \left\langle \Psi_N, \Big( \sum_{j=1}^N h_j \Big)^2 \Psi_N \right\rangle \le  \frac{C(1+|e_{N,\eps}|)^2}{\eps^2} 
\end{equation}
and the second inequality in~\eqref{eq:one-two body bound} follows by definition of the two-body density matrix.
\end{proof}

\section{Lower bound via de Finetti: Proof of Lemma~\ref{lem:GSE-2}} \label{sec:GSE-2}

Again we can assume without loss of generality that $V \ge 1$, and hence $h\ge 1$. Take an arbitrary wave function $\Psi_N\in \gH^N$. We have
$$
\frac{\langle \Psi_N, H_N \Psi_N \rangle}{N} =  \Tr\left(K_2 \gamma_{\Psi_N}^{(2)}\right) \quad \text{where}\quad K_2= \frac{1}{2} \Big( h_x + h_y  +  w_N(x-y)\Big).
$$
Let $\Psi_N$ be a many-body wave function and $d\mu_{\Psi_N}$ the associated de Finetti measure defined in Lemma~\ref{thm:deF-measure} with the projector $P$ as in~\eqref{eq:cut-off}. We write
\begin{align} \label{eq:Tr-KN-G2-decomposition}
 \Tr\left(K_2 \gamma_{\Psi_N}^{(2)}\right) &= \int \langle u^{\otimes 2}, K_2 u^{\otimes 2} \rangle d\mu_{\Psi_N}(u) + \Tr(K_2 (\gamma_{\Psi_N}^{(2)} - P^{\otimes 2}\gamma_{\Psi_N}^{(2)}P^{\otimes 2})) \nn \\
&+ \Tr \left(K_2 \left(P^{\otimes 2}\gamma_{\Psi_N}^{(2)}P^{\otimes 2} - \int |u^{\otimes 2}\rangle \langle u^{\otimes 2}| d\mu_{\Psi_N}(u)\right) \right)
\end{align}
and bound the right side from below term by term. 

\medskip

\noindent {\bf Main term.} By the variational principle we have 
\begin{align}\label{eq:validity-Hartree-A0} \int \langle u^{\otimes 2}, K_2 u^{\otimes 2} \rangle d\mu_{\Psi_N}(u) &= \int \cE_{{\rm H},N}(u)  d\mu_{\Psi_N}(u) \ge e_{{\rm H},N} \int d\mu_{\Psi_N}.
\end{align}
On the other hand, using \eqref{eq:int-muN} and $Q\le L^{-1}h$ with $ Q := \1 - P$, we have
\begin{align*}
\int d\mu_{\Psi_N} & \ge \left(\Tr\left(P\gamma_{\Psi_N}^{(1)}\right)\right)^2 = \left( 1 - \Tr \left(Q\gamma_{\Psi_N}^{(1)}\right) \right)^2  \\
&\ge 1 - 2 \Tr \left(Q\gamma_{\Psi_N}^{(1)}\right)  \ge  1 -2L^{-1} \Tr\left(h\gamma_{\Psi_N}^{(1)} \right).
\end{align*}
Since $|e_{{\rm H},N}|\le C$, \eqref{eq:validity-Hartree-A0} reduces to 
\begin{align} \label{eq:decompositon-A0} \int \langle u^{\otimes 2}, K_2 u^{\otimes 2} \rangle d\mu_{\Psi_N}(u) \ge e_{{\rm H},N} - CL^{-1} \Tr(h\gamma_{\Psi_N}^{(1)} ) .
\end{align}

\medskip

\noindent {\bf First error term.} Using $Ph \le LP$ and Lemma~\ref{thm:deF-measure} we find that
\begin{equation*} 
\left| \Tr \left((h_1+h_2) \left(P^{\otimes 2}\gamma_{\Psi_N}^{(2)}P^{\otimes 2} - \int |u^{\otimes 2}\rangle \langle u^{\otimes 2}| d\mu_{\Psi_N}(u)\right) \right) \right| \leq CL \frac{d}{N}.
\end{equation*}
On the other hand, using Equation~\eqref{eq:W<=Delta-Delta}, we have 
$$P^{\otimes 2}|w_N(x_1-x_2)|P^{\otimes 2}\le C_\delta ((Ph)_1 \otimes (Ph)_2)^{1/2+\delta} \le CL^{1+2\delta}P^{\otimes 2}$$
for all $\delta>0$. Therefore, using Lemma~\ref{thm:deF-measure} again, we find 
\begin{equation*} 
\left|\Tr \left(w_N \left(P^{\otimes 2}\gamma_{\Psi_N}^{(2)}P^{\otimes 2} - \int |u^{\otimes 2}\rangle \langle u^{\otimes 2}| d\mu_{\Psi_N}(u)\right) \right) \right| \leq C_\delta L^{1+2\delta} \frac{d}{N}.
\end{equation*}
Thus for all $\delta>0$,
\begin{equation} \label{eq:decompositon-A1}
\Tr \left(K_2 \left(P^{\otimes 2}\gamma_{\Psi_N}^{(2)}P^{\otimes 2} - \int |u^{\otimes 2}\rangle \langle u^{\otimes 2}| d\mu_{\Psi_N}(u)\right) \right) \ge -C_\delta L^{1+2\delta} \frac{d}{N}.
\end{equation}

\medskip

\noindent {\bf Second error term.} Since $h$ commutes with $P$ and $h\ge hP$, we have
\begin{align*}
\Tr \Big( (h_1+h_2)(\gamma_{\Psi_N}^{(2)}-P^{\otimes 2}\gamma_{\Psi_N}^{(2)}P^{\otimes 2})\Big) 
=  \Tr \Big( \Big[ (h_1+h_2) - P^{\otimes 2}(h_1+h_2)P^{\otimes 2} \Big] \gamma_{\Psi_N}^{(2)} \Big) \ge 0.
\end{align*}
Using the Cauchy-Schwarz inequality for operators
$$
\pm (AB+B^*A^*) \le \eta^{-1} AA^* + \eta B^*B, \quad \forall \eta>0,
$$
we find that
\begin{align*}
&\pm 2\left(\gamma_{\Psi_N}^{(2)} - P^{\otimes 2}\gamma_{\Psi_N}^{(2)}P^{\otimes 2}\right) \\
 &= \pm \Big[ (1 - P^{\otimes 2}) \gamma_{\Psi_N}^{(2)} + \gamma_{\Psi_N}^{(2)} (1-P^{\otimes 2}) + P^{\otimes 2}  \gamma_{\Psi_N}^{(2)}  (1 - P^{\otimes 2}) +  (1 - P^{\otimes 2}) \gamma_{\Psi_N}^{(2)} P^{\otimes 2} \Big] \\
& \le 2\eta^{-1} (1 - P^{\otimes 2})  \gamma_{\Psi_N}^{(2)} (1 - P^{\otimes 2}) + \eta( \gamma_{\Psi_N}^{(2)} + P^{\otimes 2} \gamma_{\Psi_N}^{(2)} P^{\otimes 2} )
\end{align*}
for all $\eta>0$. Taking the trace against $(w_N)_\pm$ and optimizing over $\eta>0$ we find that 
\begin{align} \label{eq:KG-PGP-CS}
\Tr(w_N(\gamma_{\Psi_N}^{(2)} - P^{\otimes 2} &\gamma_{\Psi_N}^{(2)}P^{\otimes 2})) \ge  - \sqrt{2}\Big( \Tr \Big(|w_N| (\gamma_{\Psi_N}^{(2)} + P^{\otimes 2}  \gamma_{\Psi_N}^{(2)} P^{\otimes 2} ) \Big)\Big)^{1/2} \nn\\
&\times \Big( \Tr \Big(|w_N| (1 - P^{\otimes 2})  \gamma_{\Psi_N}^{(2)} (1 - P^{\otimes 2}) \Big) \Big)^{1/2}.
\end{align}
Using again Equation~\eqref{eq:W<=Delta-Delta} and the elementary fact
\begin{align} \label{eq:t^r=}
t^{r}=\inf_{\eta>0}\Big(r\eta^{-1}t + (1-r) \eta^{r/(1-r)}\Big)\quad \text{for all}~~ t\ge 0, r\in (0,1)
\end{align}
we get
$$ |w_N(x-y)| \le C_\delta (h_xh_y)^{1/2+\delta} \le C_\delta( \eta^{-1} h_xh_y + \eta^{\frac{1+2\delta}{1-2\delta}}) \quad \text{for all}~~ \delta \in (0,1/2), \eta>0.$$
Taking the trace against $\gamma_{\Psi_N}^{(2)} + P^{\otimes 2}  \gamma_{\Psi_N}^{(2)} P^{\otimes 2}$ and optimizing over $\eta>0$ (cf. \eqref{eq:t^r=}), we get
$$ \Tr \left( |w_N|\left(\gamma_{\Psi_N}^{(2)} + P^{\otimes 2}  \gamma_{\Psi_N}^{(2)} P^{\otimes 2} \right) \right) \le 2 C_\delta  \left( \Tr \left( h\otimes h \gamma_{\Psi_N}^{(2)}  \right) \right) ^{1/2+\delta}
$$
for all $\delta\in (0,1/2)$. Similarly, from \eqref{eq:W<=Delta-Delta}, \eqref{eq:t^r=} and $Q\le L^{-1}h$, we find that 
\begin{align*}
(1-P^{\otimes 2}) |w_N| (1-P^{\otimes 2}) &\le C_\delta (1-P^{\otimes 2}) (h\otimes h)^{1/2+\delta} (1-P^{\otimes 2}) \\
& \le C_\delta \Big[ (Q h^{1/2+\delta} ) \otimes h^{1/2+\delta} + h^{1/2+\delta} \otimes (Qh^{1/2+\delta} ) \Big] \\
& \le \frac{C_\delta}{L^{1/2-\delta}} \Big[ h \otimes h^{1/2+\delta} + h^{1/2+\delta} \otimes h \Big] \\
& \le \frac{C_\delta}{L^{1/2-\delta}} \Big[ \eta^{-1} h \otimes h + \eta^{\frac{1+2\delta}{1-2\delta}}(h\otimes 1+ 1\otimes h)\Big]
\end{align*}
for all $\delta\in (0,1/2)$ and $\eta>0$. Taking the trace against $\gamma_{\Psi_N}^{(2)}$ and optimizing over $\eta>0$ 
we deduce that
\begin{align*}
\Tr \Big( |w_N|(1-P^{\otimes 2}) \gamma_{\Psi_N}^{(2)} (1-P^{\otimes 2}) \Big) \le \frac{C_\delta}{L^{1/2-\delta}}  \left( \Tr \Big(h \gamma_{\Psi_N}^{(1)} \Big)  \right)^{1/2-\delta} \left( \Tr \Big( h\otimes h \gamma_{\Psi_N}^{(2)}  \Big) \right) ^{1/2+\delta}
\end{align*}
for all $\delta\in (0,1/2)$. Therefore, ir follows from \eqref{eq:KG-PGP-CS} that
\begin{align} \label{eq:decompositon-A2}
&\Tr(w_N(\gamma_{\Psi_N}^{(2)} - P^{\otimes 2} \gamma_{\Psi_N}^{(2)}P^{\otimes 2})) \nn\\
&\ge  - \frac{C_\delta}{L^{1/4-\delta/2}} \left( \Tr \Big(h \gamma_{\Psi_N}^{(1)} \Big)  \right)^{1/4-\delta/2} \left( \Tr \Big( h\otimes h \gamma_{\Psi_N}^{(2)}  \Big) \right) ^{1/2+\delta}.
\end{align}

\medskip

\noindent {\bf Summary.} Inserting the estimates \eqref{eq:decompositon-A0}, \eqref{eq:decompositon-A1} and \eqref{eq:decompositon-A2} in \eqref{eq:Tr-KN-G2-decomposition} we find the desired lower bound. \qed

\section{Final energy estimate: Proof of Theorem \ref{thm:cv-nls}} \label{sec:GSE-4}

We again assume, without loss of generality, that $V \ge 1$. We apply Lemma~\ref{lem:GSE-2} to a ground state $\Psi_N$ of $H_N$, then insert the dimension estimate~\eqref{eq:estimate-d-dimP} and the results of Lemma~\ref{lem:GSE-3} (recall the definition~\eqref{eq:ener epsilon}). This gives
\begin{align} \label{eq:validity-Hartree-last-estimate}
e_{{\rm H},N}\ge e_N \ge e_{{\rm H},N} &- C_{\delta} \left(  \frac{L^{2+2/s+2\delta}}{N} + \frac{1}{L^{1/4-\delta/2}} \left( \frac{1+|e_{N,\eps}|}{\eps}\right) ^{5/4+3\delta/2} \right) 
\end{align}
for all $\eps>0$, $\delta\in (0,1/2)$, $N \ge 2$ and $L \ge 1$. 

\medskip

\noindent{\bf Stability of the second kind.} We will deduce from \eqref{eq:validity-Hartree-last-estimate} that $|e_{N,\eps}| \leq C$ for $\eps>0$ small, provided~\eqref{eq:beta condition} holds. Using \eqref{eq:validity-Hartree-last-estimate} with $w$ replaced by $(1-\eps)^{-1}w$, we have
\begin{align} \label{eq:validity-Hartree-last-estimate-eps}
e_{{\rm H},N}^{\eps}\ge e_{N,\eps} \ge e_{{\rm H},N}^{\eps} &- C_{\delta} \left(  \frac{L^{2+2/s+2\delta}}{N} + \frac{1}{L^{1/4-\delta/2}} \left( \frac{1+|e_{N,{\eps'}}|}{\eps'-\eps}\right) ^{5/4+3\delta/2} \right) 
\end{align}
for all $1>\eps'>\eps>0$ and $\delta\in (0,1/2)$, where $e_{{\rm H},N} ^{\eps}$ is the ground state energy of the Hartree functional with $h$ replaced by $(1-\eps)h$ (similarly as in~\eqref{eq:ener epsilon}). 
Using Assumption~\eqref{eq:assumption-w2}, Lemma~\ref{lem:Hartree-NLS} and the diamagnetic inequality $\langle u, h u \rangle \ge \int |\nabla |u||^2$, we find that there exists some $\eps_0>0$ (depending only on $w$) such that 
\begin{align} \label{eq:eHN-eps} e_{{\rm H},N} ^{\eps} \geq -C ~~\mbox{ for all } 0<\eps<\eps_0.
\end{align}
We make the induction hypothesis (labeled $I_\eta$)
\begin{equation}\label{eq:induc}
\limsup_{N\to \infty} \frac{|e_{N,\eps}|}{1+N^{\eta}} <\infty \mbox{ for all } 0< \eps <\eps_0. 
\end{equation}
Note that $I_\eta$ holds for $\eta = 2 \beta -1$ by Lemma~\ref{lem:GSE-1}, and we ultimately aim at proving $I_0$. From \eqref{eq:validity-Hartree-last-estimate} and $\eqref{eq:eHN-eps}$, by choosing $L=N^\tau$ with $\tau>0$, we deduce that if $I_\eta$ holds for some $\eta \leq 2\beta - 1$, then $I_{\eta'}$ also holds provided that 
\begin{align} \label{eq:eta'-eta} \eta' >  \max \Big\{ \tau(2+2/s) - 1, (5\eta-\tau)/4 \Big\}\quad \text{for some}~\tau>0.
\end{align}
With the optimal choice $\tau=s(5\eta+4)/(9s+8)$, the requirement \eqref{eq:eta'-eta} reduces to
\begin{align} \label{eq:eta'-eta-b} \eta' > \eta - \frac{s-\eta(s+2)}{9s+8}.
\end{align}
When $\beta<(s+1)/(s+2)$, we can choose a constant $c$ such that
$$ 0< c < \frac{s-(2\beta-1)(s+2)}{9s+8}$$
and it is clear that \eqref{eq:eta'-eta-b} holds with $\eta'=\eta-c$ because $\eta\le 2\beta-1$. Thus we have shown that $I_{\eta}$ implies $I_{\eta - c}$ for some constant $c>0$ independent of $\eta$. Repeating the argument sufficiently many times we finally deduce that $I_0$ holds, which is the desired stability bound. 

\medskip

\noindent{\bf Conclusion.} Now, using $|e_{N,\eps}|\le C$ for $\eps>0$ small, \eqref{eq:validity-Hartree-last-estimate} reduces to 
\begin{align} \label{eq:validity-Hartree-net-1}
e_{{\rm H},N}\ge e_N \ge e_{{\rm H},N} &-  C_\delta \Big( \frac{L^{2+2/s+2\delta}}{N} + \frac{1}{L^{1/4-\delta/2}} \Big)
\end{align}
for all $\delta \in (0,1/2)$ and $L\ge 1$. By choosing $L= N^{4/(9+8/s)}$ we conclude that
\begin{align} \label{eq:validity-Hartree-net-2}
e_{{\rm H},N}\ge e_N \ge e_{{\rm H},N} - C_\alpha N^{-\alpha}
\end{align}
for every $0<\alpha<s/(9s+8)$. The desired energy convergence~\eqref{eq:thm-cv-GSE} follows from~\eqref{eq:validity-Hartree-net-2} and $\lim_{N\to \infty} e_{{\rm H},N} = e_{\rm NLS}$ (see Appendix \ref{sec:app}). 
Once the convergence of the energy is established, the convergence of states~\eqref{eq:thm-cv-DM} follows exactly as in~\cite[Section 4.3]{LewNamRou-14c}, and we omit the details.

\appendix

\section{From Hartree to NLS}\label{sec:app}

Here we prove an elementary lemma which, together with the variational principle, implies that $\lim_{N\to \infty}e_{{\rm H},N}=e_{\rm NLS}$. 

\begin{lemma}[\textbf{Limit of the Hartree interaction energy}] \label{lem:Hartree-NLS}\mbox{}\\
For every $w\in L^1(\R^2)$ with $\int w=a$,   
\begin{align*}
\lim_{\lambda\to \infty} \sup_{\substack{u\in H^1\\ u\ne 0}}  \left\| |u| \right\|_{H^1}^{-4}\left|\iint |u(x)|^2 \lambda^3  w(\lambda (x-y)) |u(y)|^2 dxdy - a\int |u(x)|^4dx \right| =0.
\end{align*}
\end{lemma}

\begin{proof} It suffices to consider the case when $u\ge 0$. By introducing the variable $z=\lambda(x-y)$, we can write
\begin{align}
&\iint |u(x)|^2 \lambda^3  w(\lambda (x-y)) |u(y)|^2 dxdy - a\int |u(x)|^4dx\nn\\ 
=& \iint |u(x)|^2  w(z) \Big( |u(x-\lambda^{-1}z)|^2 - |u(x)|^2 \Big) dxdz. \label{eq:Hartree-NLS-0}
\end{align} 
Now we pick $L>0$ and decompose 
$$w(z)=\1(|z|> L)w(z)+ \1(|z|\le L)w(z).$$ 
By the Cauchy-Schwarz inequality 
$$2|u(x)|^2 |u(x-\lambda^{-1}z)|^2 \le |u(x)|^4 + |u(x-\lambda^{-1}z)|^4$$
and the Sobolev's embedding $\|u\|_{L^4}\le C \|u\|_{H^1}$ we have
\begin{align}
&\left| \iint |u(x)|^2  \1(|z|> L)w(z) \Big( |u(x-\lambda^{-1}z)|^2 - |u(x)|^2 \Big) dxdz \right| \nn\\
 \le &  \iint \Big( \frac{3}{2}|u(x)|^4 + \frac{1}{2} |u(x-\lambda^{-1}z)|^4 \Big)  \1(|z|> L)|w(z)| dxdz \nn\\
 = & 2 \|u\|_{L^4}^4 \int_{|z|>L}|w(z)| dz \le C \|u\|_{H^1}^4 \int_{|z|>L}|w(z)| dz. \label{eq:Hartree-NLS-1}
\end{align}  
On the other hand, note that
\begin{align*} &\left| |u(x-\lambda^{-1}z)|^2 - |u(x)|^2 \right|=  \left| \int_0^1 (-\lambda^{-1}z).(\nabla |u|^2)(x-t \lambda^{-1}z) dt  \right| \\
\le & 2 \lambda^{-1}|z| \int_0^1 |\nabla u(x-t \lambda^{-1}z)|.|u(x-t \lambda^{-1}z)|  dt 
\end{align*}
where we have used $|\nabla (u^2)| \le 2|\nabla u|.|u|$ in the last estimate. Combining with Fubini's theorem and Sobolev's inequality $\|u\|_{L^6}\le C \|u\|_{H^1}$, we find that
\begin{align}
&\left| \iint |u(x)|^2  \1(|z|\le L)w(z) \Big( |u(x-\lambda^{-1}z)|^2 - |u(x)|^2 \Big) dxdz \right| \\
\le & 2\lambda^{-1}L  \int_0^1 \int |w(z)| \left( \int |u(x)|^2 |\nabla u(x-\lambda^{-1}z)|. |u(x-\lambda^{-1}z)| dx \right) dz dt \nn\\
\le &  2\lambda^{-1}L  \int_0^1 \int |w(z)| \Big(\int |u(x)|^6 dx \Big)^{1/3} \Big(\int |\nabla u(x-\lambda^{-1}z)|^2 dx \Big)^{1/2} \times \nn\\
& \quad\quad\quad\quad\quad\quad\quad\quad\quad\quad\quad\quad\quad\quad\quad\times \Big( \int |u(x-\lambda^{-1}z)|^6 dx \Big)^{1/6} dz dt  \nn\\
\le &  C \lambda^{-1}L \|u\|_{H^1}^4 \|w\|_{L^1}. \label{eq:Hartree-NLS-2}
\end{align}  
From \eqref{eq:Hartree-NLS-0}, \eqref{eq:Hartree-NLS-1} and \eqref{eq:Hartree-NLS-2}, it follows that 
\begin{align*}
& \|u\|_{H^1}^{-4} \left| \iint |u(x)|^2 \lambda^3  w(\lambda (x-y)) |u(y)|^2 dxdy - a\int |u(x)|^4dx \right|\nn\\ 
\le & C  \left( \int_{|z|>L}|w(z)| + \lambda^{-1}L \|w\|_{L^1} \right).
\end{align*} 
The conclusion follows by choosing $1\ll L\ll \lambda$ (for example $L=\sqrt{\lambda}$).
\end{proof}

\end{document}